\documentclass[journal,onecolumn,12pt]{IEEEtran}
\usepackage{epsfig,latexsym}
\usepackage{float}
\usepackage{indentfirst}
\usepackage{amsmath}
\usepackage{amssymb}
\usepackage{times}
\usepackage{subfigure}
\usepackage{cite}
\usepackage{url}
\newtheorem{Lemma}{Lemma}

\newtheorem{Remark}{Remark}

\newtheorem{theorem}{$\mathbf{Theorem}$}

\begin{document}
\title{On the Application  of Noisy Network Coding to the  Relay-Eavesdropper Channel }
\author{Peng Xu, Zhiguo Ding, \IEEEmembership{Member, IEEE}, Xuchu
Dai, and Kin Leung, \IEEEmembership{Fellow, IEEE}
\thanks{Peng Xu and Xuchu Dai are with Dept. of Electronic Engineering and Information Science
, University of Science and Technology of China, P.O.Box No.4,
230027, Hefei, Anhui, China. Zhiguo Ding is with  School of
Electrical, Electronic, and Computer Engineering Newcastle
University,  UK. Kin Leung is with Department of Electrical and
Electronic Engineering, Imperial College, London, UK.

 Submitted to IEEE Transaction on Information Theory at 14 March
 2012.}} \maketitle \vspace{-3em}

\vspace{1em}
\begin{abstract}
  In this paper,  we consider the design of a new secrecy transmission scheme    for a four-node relay-eavesdropper channel. The key idea of the   proposed scheme is to combine   noisy network coding   with the  interference assisted strategy for wiretap channel with a helping interferer. A new   achievable secrecy rate is characterized  for both discrete memoryless and  Gaussian channels. Such a new rate can be viewed as a general framework, where the existing interference assisted schemes such as noisy-forwarding  and cooperative jamming approaches can be shown as special cases of the proposed scheme. In addition, under some channel condition where
    the existing schemes can only achieve zero secrecy rate, the proposed secrecy scheme can still offer significant performance gains.

\end{abstract}
\section{introduction}
Cooperation and secrecy are two important concepts which have  been widely studied in wireless communications. Due to the broadcast nature of radio propagation, wireless transmissions can be over-heard by multiple unintended  receivers. Such broadcasting nature facilitates cooperation by allowing neighbor  users to intelligently exploit the over-heard information but also leads to a serious security problem such as eavesdropper attacking.

Wyner studied the eavesdropping attack from an information theoretic aspect by introducing the concept of wire-tap channel in \cite{Wyner1975}. Under the assumption that the wiretapper channel (source-to-wiretapper) is a degraded version of the main channel (source-to-destination), the secrecy capacity was established based on the rate-equivocation region concept. Csiz{\' a}r and K\"{o}rner extended this degraded channel to the general wiretapper channel setup without any special assumptions, and found the secrecy capacity in \cite{csisz¨¢r1978broadcast}. Recently, more types of multiuser networks have been studied in the context of secrecy communications. The multiple access  wiretapper channel (MAC-WT)
is considered in \cite{tang2007multiple,tekin2008general}, where  an external passive wiretapper was included. The multiple access channel with confidential messages (MAC-CM) is studied in   \cite{liu2006discrete,liang2008multiple}, where the source terminals act as eavesdroppers to each other. The works in \cite{liu2008discrete,liu2009secrecy} investigate the broadcast channel with confidential messages (BC-CM) where both receivers wish to keep their message secret from the others. Similarly, \cite{liu2008discrete,koyluoglu2011interference} considered the interference channel (IC), where the former treated each unintended receiver as an eavesdropper and the latter introduced an external eavesdropper.

 To further enhance the secrecy level in the above mentioned channel models, user cooperation has been considered in \cite{oohama2006relay,he2010cooperation,he2009two,ekrem2008effects,
 ekrem2011secrecy, tekin2008general,liu2008discrete,lai2008relay,
 tang2011interference,koyluoglu2011cooperative}, where \cite{oohama2006relay,he2010cooperation,he2009two,ekrem2008effects,
 ekrem2011secrecy} consider the untrusted helper scenario in which the relay node acts both a helper and an eavesdropper, and \cite{tekin2008general,liu2008discrete,lai2008relay,
 tang2011interference,koyluoglu2011cooperative} consider cooperative communication systems with an external eavesdropper. Particularly for the latter scenario with a dedicated relay, such as the work in \cite{tekin2008general,lai2008relay,tang2011interference} where the relay does not have its own message to be sent or received, so-called {\em interference-assisted} schemes that involve cooperative jamming \cite{tekin2008general} and noisy forward (NF) \cite{lai2008relay} techniques have  been proved to be particularly useful to increase secrecy. The basic idea of these strategies is to allow the relay to send codewords (or even pure noisy) which are {\em independent} to the source message in order to {\em confuse} the eavesdropper. The work in\cite{tang2011interference} can be viewed as a generalization of this type of interference strategies for a wiretap channel with a helping interferer (WT-HI). On the other hand, \cite{lai2008relay} extended some classical relaying schemes for relay channels (such as Cover and El Gamal's decode-and-forward (DF) and compress-and-forward (CF) schemes \cite{cover1979capacity}) to {\em strength} the main channel.   It has been well known that, in comparison with DF, CF with independent coding at the relay is more suitable for a general scenario without a strong source-relay link. However, for the CF scheme  in \cite{lai2008relay}, its lower bound on the equivocation rate is the same as the NF scheme. That is to say, the source  information forwarded from the relay to the destination was shown in \cite{lai2008relay} not helpful to this four-node relay-eavesdropper channel compared to interference-assisted schemes.

In this paper, we consider the four-node secrecy communication scenario with a source, a relay, a destination and an eavesdropper.
 The aim of this paper is to demonstrate that an effective use of the relay for source information forwarding can yield a larger achievable secrecy rate than the interference assisted schemes, whereas  the  existing relay secrecy protocol in \cite{lai2008relay} can only achieve the same performance as the interference assisted ones.   The key idea of our achievable scheme is to combine noisy network coding (NNC) \cite{lim2011noisy} for relay channels with the interference-assisted scheme \cite{tang2011interference} for WT-HI. The proposed scheme is mainly facilitated  by the fact that  NNC can be used as an efficient tool to analyze the achievable rates of a large scale network, which is particularly useful for the addressed four-node secrecy network. But slightly different to decoding in \cite{lim2011noisy}, which did not involve uniquely decoding the relay message,  we discuss  how to decode the relay messages at the eavesdropper when the equivocation rate is commutated.   As a result, we can obtain a larger achievable secrecy rate with more explicit expressions if compared to the traditional CF scheme in \cite{lai2008relay}. Furthermore, following stochastic encoding in \cite{Wyner1975,csisz¨¢r1978broadcast}, random dummy information has been blended into the encoding procedure  at the both   source and   relay, which ensures an effective application of NNC in secrecy communications.

After the achievable rate of the proposed scheme is obtained, it is shown that the interference-assisted scheme in \cite{tang2011interference} can be viewed as a special case of our proposed scheme when we ask the relay only to send information not related to the source message. When the relay-destination link is strong, we will show that the proposed scheme can  exploit the capability of the dedicated relay node more effectively. That is to say, in this case, the relay can strength the condition of the main channel and suppress the wiretapper channel at the same time. It is interesting to observe that  for some very strong eavesdropping cases, the proposed scheme can still achieve a positive secrecy rate, while the achievable secrecy rate of the scheme in \cite{tang2011interference} is zero.
 To further visualize the impact of our achievable cooperative scheme on secrecy, we extend the achievable secrecy rate to the   Gaussian relay-eavesdropper channel based on Gaussian codebooks. Fixed power control is first used to obtain an explicit expression of the achievable rate, and the impact of power control on the secrecy rate is investigated by using computer simulations. Particularly the provided numerical results demonstrate that the use of power control can yield more performance gains for  the proposed scheme compared to the  scheme in \cite{tang2011interference}. For example, in the case of very strong eavesdropping to which the scheme in \cite{tang2011interference} can only achieve zero secrecy rate, a positive secrecy rate can be achieved with a moderate requirement of the relay-destination channel condition.

 The reminder of the paper is organized as following. Section II describes the channel model of the addressed relay-eavesdropper scenario. Section III states the achievable rate of the proposed secrecy transmission protocol, and some remarks by comparing the obtained rate to existing ones. Section IV provides the extension of the secrecy rate to the Gaussian relay-eavesdropper channel, where the impact of the power control will also be discussed. Conclusions are given in Section V. Proofs are collected in Appendix. Throughout this paper, $(x)^+$ denotes $\max(0,x)$. A sequence of random variables with  time index $i\in[1:n]:=\{1,\cdots,n\}$ is denoted as $X^n:=\{X_1,\cdots,X_n\}$. And $X_{]i[}:=\{X_j,1\leq j<i \textrm{ or } i<j\leq n\}$.

\section{The relay-eavesdropper channel}
\begin{figure}[htbp]\centering
    \epsfig{file=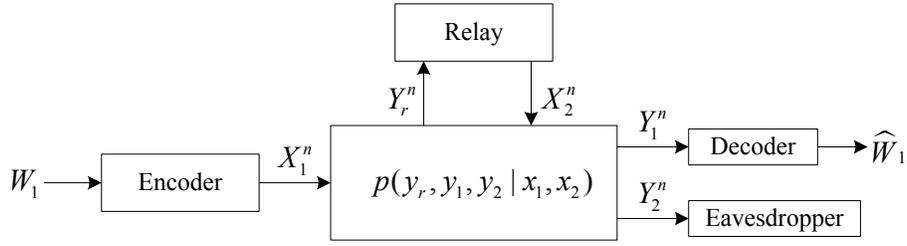, width=0.7\textwidth,clip=}
\caption{The relay eavesdropper  channel.}\label{system channel}
\end{figure}
 Consider a discrete memoryless relay-eavesdropper channel with a source $(X_1)$, a dedicated relay $(Y_r,X_2)$, a destination $(Y_1)$ and a passive eavesdropper $(Y_2)$. This communication model consists of two finite input alphabets $\mathcal{X}_1$, $\mathcal{X}_2$ at the source and relay respectively, three output alphabets $\mathcal{Y}_r$, $\mathcal{Y}_1$, $\mathcal{Y}_2$ at the relay, destination and eavesdropper respectively, and a channel transition probability distribution $p(y_r,y_1,y_2|x_1,x_2)$ where $x_t\in \mathcal{X}_t$, $y_t\in \mathcal{Y}_t$ $(t=1,2)$ and $y_r\in \mathcal{Y}_r$. The source wishes to send a confidential message $W_1$ to the destination with the help of the dedicated relay, while keeping it secret from the eavesdropper which knows the codebooks of the source and relay. We refer to such a cooperative communication model as the {\em relay-eavesdropper channel}, as shown in Fig. \ref{system channel}.

 The source intends to send a confidential  message $W_1\in \{1,\cdots,M\}$ to the destination in $n$ channel uses. The memoryless assumption is imposed  in the sense that at the $i$-th channel use  the channels outputs $(y_{r,i}, y_{1,i},y_{2,i})$ only depend on the channel inputs $(x_{1,i},x_{2 ,i})$. A stochastic encoder for the source is specified by a matrix of conditional probabilities $f_1(x_{1,i}|w_1)$, where $x_{1,i}\in \mathcal {X}_1$, $w_1 \in \mathcal{W}_1$ and $\sum_{x_{1,i}\in \mathcal{X}_1} f_1(x_{1,i}|w_1)=1$ for all $i=1,\cdots,n$. The encoder at the relay maps the signals $(y_{r,1},y_{r,2},\cdots,y_{r,i-1})$ received before the $i$-th channel use to its output $x_{2,i}$, using another stochastic encoder $f_2$, which is described by a matrix of conditional probabilities $f_2 (x_{2,i}|y^{i-1}_{r}, x_2^{i-1})$, where $x_{2,k}\in \mathcal{X}_2,(k=1,\cdots,i)$, $y_{r,k}\in \mathcal{Y}_r ,(k=1,\cdots,i-1)$ and $\sum_{x_{2,i}\in \mathcal{X}_2} f_2 (x_{2,i}|y^{i-1}_{r}, x_2^{i-1}) =1$. The decoding function at the destination is described by a deterministic mapping $\phi$: $\mathcal{Y}_1^n\rightarrow \mathcal{W}_1$. The average error probability of a $(M,n)$ code is \begin{eqnarray}
   P_e^{(n)} =\frac{1}{M} \sum_{w_1=1}^{M} Pr\{\phi(Y_1^n)\neq w_1 |w_1 \textrm{ was sent}\}.
 \end{eqnarray}

 The secrecy level at the eavesdropper is measured with respect to the equivocation rate $\frac{1}{n}H(W_1|Y_2^n)$. A  secrecy rate $R_s$ is said to be {\em achievable} for the relay-eavesdropper channel if for any $\epsilon>0$ there exists a sequence of codes $(M,n)$ such that \begin{eqnarray}
   M\geq 2^{nR_s}, P_e^{(n)}\leq \epsilon,\textrm{ and}\nonumber\\
   R_s-\epsilon\leq \frac{1}{n}H(W_1|Y_2^n)
 \end{eqnarray}
 for sufficiently large $n$.
\section{An achievable scheme}
In this section we will present an achievable scheme for the relay-eavesdropper channel, which combines the NNC scheme for relay channel \cite{lim2011noisy}, the random binning scheme for wiretap channels \cite{Wyner1975}, \cite{csisz¨¢r1978broadcast}, and the interference assisted scheme for WT-HI  \cite{tang2011interference}. As shown at the end of this section, the achievable rate obtained in   \cite{tang2011interference} can be viewed as a special case of the results obtained in this paper, by asking the relay to ignore the over-heard source information and only to transmit dummy messages.   The    equivocation rate achieved by the proposed secrecy scheme is given in the following subsection.
\subsection{Achievable secrecy rate}
Prior to the discussions  of the achievable rate,
we first give some definitions as:
\begin{align}&I_1=I(\hat{Y}_r;Y_r|X_2), \nonumber\\
&I_2^{(t)} =I(X_1,X_2;Y_t)+I\left(\hat{Y}_r;X_1,Y_t|X_2\right), \label{defenition}\end{align} where $(t=1,2)$, and  $R_1^{(t)}$ is defined as a function of  $R_2$ as
\begin{eqnarray}\label{R_1(R_2)} R_1^{(t)}(R_2) =\max\left\{ \min\left[I\left(X_1;\hat{Y}_r,Y_t|X_2\right), I_2^{(t)}-R_2 \right], I(X_1;Y_t) \right\} .\end{eqnarray}
Then our achievable secrecy rate is given by the following theorem.
\begin{theorem}
The achievable secrecy rate $R_s$ for the addressed relay-eavesdropper channel is
\begin{eqnarray}\label{simple}
  R_s=\max_{\pi,R_2\geq I_1} \left[R_1^{(1)}(R_2)-R_1^{(2)}(R_2)\right]^+,
\end{eqnarray}
 where $\pi$ denotes the class of distributions
$$p(x_1)p(x_2)p(y_r,y_1,y_2|x_1,x_2)p(\hat{y}_r|y_r,x_2) .$$
\end{theorem}
\begin{proof}
  refer to Appendix A.
\end{proof}

To achieve the secrecy rate given in Theorem 1, the source and the relay  cooperate with each other based on the idea of the NNC scheme, and the necessary \emph{randomness} will be blended into the codebooks at both   two transmitters in order to increase the secrecy level. Roughly speaking, a message is first generated at the source by mixing the confidential message   with a random dummy message, where the redundant dummy part provides  randomization in order to ensure the confidential part transmitted under the secrecy constraint. Then the source encodes this mixed message and sends it multiple times in multiple blocks using independent codebooks, which is different to the classical CF scheme in which  a source message is first divided  into multiple segments and then   transmitted over multiple blocks. Meanwhile, at each block, the relay uses the NNC strategy without Wyner-Ziv binning to send the quantized observations obtained from the previous block. For many existing works, such as \cite{lai2008relay}, the relay observations are first compressed and then random binning is used to insert randomness   into the system, where the two steps are performed over two separate stages. A key idea of the proposed secrecy scheme is to treat the size of the relay codebook as a parameter, where the dummy message is generated at the same time when the source message is compressed at the relay, and the amount of randomness injected into the system can be easily adjusted by changing this parameter.
% Surprisingly such a flexible choice of the relay codebook size can help us to obtain a better achievable rate than the schemes in \cite{lai2008relay,tang2011interference}.

For decoding, the destination performs either  special ``joint" decoding according to NNC or  simple separated decoding by ignoring the relay message. Specifically, when such  ``joint" decoding is performed, the destination utilizes the observations from the both source and relay, and simultaneously decodes the received signals from all the blocks without decoding the compression indices. Separate decoding simply aims to decode the source message by treating the signals transmitted by the relay as pure noise. Some remarks of the achievable scheme are given as follows.
\begin{Remark}
  In Theorem 1,  the first factor of $R_1^{(1)}(R_2)$, $\min\left[I\left(X_1;\hat{Y}_r,Y_1|X_2\right), I_2^{(1)}-R_2 \right]$,  can be viewed as a special case of the NNC achievable rate  for the classical relay channel without security constraints \cite{lim2011noisy}; $R_1^{(2)}(R_2)$ denotes the  redundancy rate sacrificed at the source, in order to protect the confidential message and confuse the eavesdropper; $R_2$ denotes the data rate at the relay. Furthermore, for the rate $R_1^{(1)}(R_2)$ given in \eqref{R_1(R_2)}, the first term in the $\max$ function, $\min\left[I\left(X_1;\hat{Y}_r,Y_1|X_2\right), I_2^{(1)}-R_2 \right]$,   achieved by using the ``joint" decoding strategy according to NNC; on the other hand, the second term, $I(X_1; Y_1)$, is achieved by using the separate decoding strategy.
\end{Remark}
\begin{Remark}
  The separate decoding strategy  at the destination corresponds to the \emph{cooperative jamming} scheme in which the relay rate $R_2$ is so large that  both the destination and eavesdropper cannot extract any source information from the relay transmissions.   This cooperative jamming is applied to the case that the relay-eavesdropper link is stronger than the relay-destination link, then the relay can  hurts the eavesdropper more even though it attacks the destination. Note that this is different to the work in \cite{lim2011noisy} for classical relay channel without secrecy constraint, which ignores the separate decoding and  only needs to simply set $R_2$ to larger than but arbitrarily close to $I_1$.
\end{Remark}
\begin{Remark}
  As shown in the next subsection, the proposed relay scheme can achieve better performance than the interference assisted scheme in \cite{tang2011interference} which in turn has been recognized to realize  better performance than the relay scheme \cite{lai2008relay}. Direct comparison between the two relay schemes is difficult due to the complicated expressions of two achievable  rates. The reason for such a performance improvement is due to the use of NNC, where NNC
   is ideal for the analysis of large scale networks and
   therefore can facilitate the calculation of a better equivocation rate.
   %Compared to the CF scheme shown in Theorem 4 of  \cite{lai2008relay}, the difference is that the proposed scheme absorbs  the superiority  of NNC whereas the coding scheme in \cite{lai2008relay} is a transformation of Wyner-Ziv coding. As a challenge  problem, Wyner-Ziv coding is hardly to be evaluated for this relay-eavesdropper channel scenario, so that Reference \cite{lai2008relay} made some changes to the Wyner-Ziv coding and used backward decoding to lower bound equivocation, which result in an imperfect achievable secrecy rate that can not even outperform the NF scheme without considering the channel output $Y_r$. Fortunately, from the explicit form of the achievable rate in Theorem 1, one can observe that the NNC scheme utilized in our propose scheme  makes the achievability  analysis tractable  for  this relay-eavesdropper  scenario. This is benefitted from the very simple CF process at the relay without any binning operation.
\end{Remark}
\subsection{Special cases}
In this subsection, we will present several special cases to show   the superiority of the proposed achievable scheme in comparison with the    interference assisted schemes in \cite{lai2008relay, tang2011interference}.
%Since the NF scheme in \cite{lai2008relay} is just a special case of the work in \cite{tang2011interference} without separate decoding, we only need to display the scheme in \cite{tang2011interference} for example.
\subsubsection{Disable the channel output $Y_r$}
If we disable the received signal $Y_r$ at the relay by setting $\hat{Y}_r= {\O}$, the channel model reduces to WT-HI in \cite{tang2011interference}, and the achievable rate becomes
\begin{eqnarray}
  R_{s,[HI]}=\max_{p(x_1)p(x_2),R_2\geq 0} \left\{ R_{1,[HI]}^{(1)}(R_2)-R_{1,[HI]}^{(2)}(R_2) \right\} \nonumber
\end{eqnarray}
where $R_{1,[HI]}^{(t)}(R_2)= \max \{\min[I(X_1;Y_t|X_2),I(X_1X_2;Y_t)-R_2], I(X_1;Y_t)$ with $t=1,2$. So   the achievable equivocation of the proposed scheme would be at least as good as the one in \cite{tang2011interference} in any case.
\subsubsection{Very strong eavesdropping}
Following  \cite{tang2011interference}, we define this case as
\begin{eqnarray}
  I(X_1;Y_2)\geq I(X_1;Y_1|X_2) \nonumber
\end{eqnarray} for all distributions of $p(x_1)p(x_2)p(\hat{y}_r|y_r,x_2)$. For such very strong eavesdropping as shown in Fig. \ref{very_strong}, the helping interferer scheme in \cite{tang2011interference} cannot achieve a positive secrecy rate as given in Eq. (14) of \cite{tang2011interference}.  However, the secrecy rate of the proposed scheme can be lower bounded as  \begin{align}\label{very-strong}R_{s, lb} =  \max_{\pi} \left\{ \min \left[ \begin{array}{l}
I(X_1; \hat{Y}_r,Y_1|X_2) -I(X_1;Y_2),\\
  I_2^{(1)}-I_2^{(2)}, \\
  I(X_1,X_2;Y_1)-I(\hat{Y}_r;Y_r|X_1,X_2,Y_1)-I(X_1;Y_2)
\end{array}\right]\right\}^+\end{align}
where the lower bound on the right side is obtained by choosing $R_2=R_2^*$ with $R_2^*=\max\{I_1,I(X_2;Y_2|X_1)+I(\hat{Y}_r;X_1,Y_2|X_2)\}$.
From Fig. \ref{very_strong}, one can see that a non-zero secrecy rate $R_s$ can still  be achievable in this case. This is because the use of the relay can strengthen  the main link and suppress the wiretapper link at the same time.
\subsubsection{Extremely strong eavesdropping}
This case is referred to $$I(X_1;Y_2) \geq \min \left\{ I(X_1;\hat{Y}_r,Y_1|X_2), I(X_1,X_2;Y_1)-I(\hat{Y}_r;Y_r|X_1,X_2,Y_1) \right\}$$ for all product distributions of $p(x_1)p(x_2)p(\hat{y}_r|y_r,x_2)$, where the term at the right side of the above equality is from \eqref{R_1(R_2)} by setting $R_2=I_1$. The proposed scheme cannot achieve any positive secrecy rate, the same as the scheme in \cite{tang2011interference}.

Note that the use of the channel prefixing technique in \cite{csisz¨¢r1978broadcast} may further enhance the performance of the proposed scheme, but we do not consider this prefixing method in this paper due to the intractable evaluation of its performance.
\begin{figure}[htb]
\begin{center}
\subfigure[Case $i$: when $I_1<I_3$] {\includegraphics[width=0.4\textwidth]{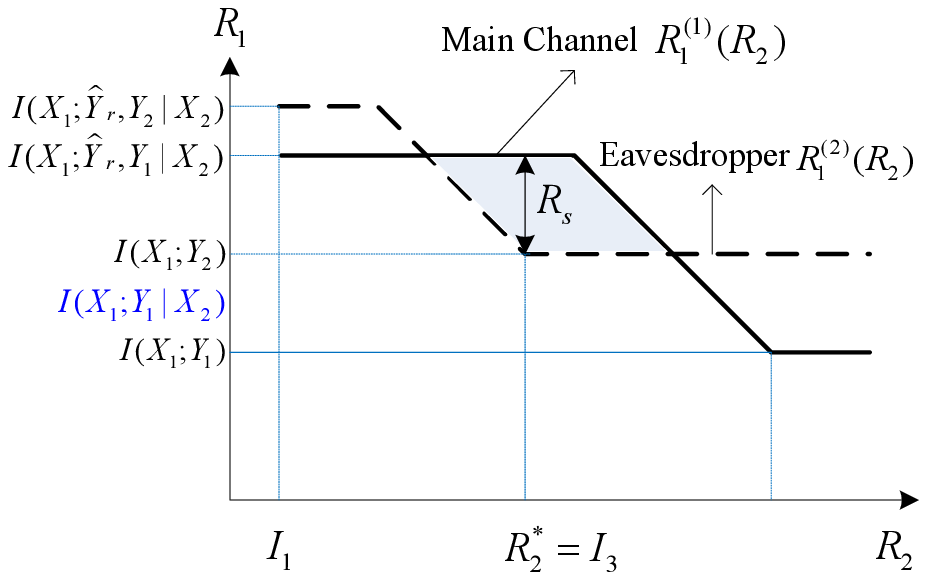}}
\subfigure[Case $i$: when $I_1\geq I_3$] { \includegraphics[width=0.4\textwidth]{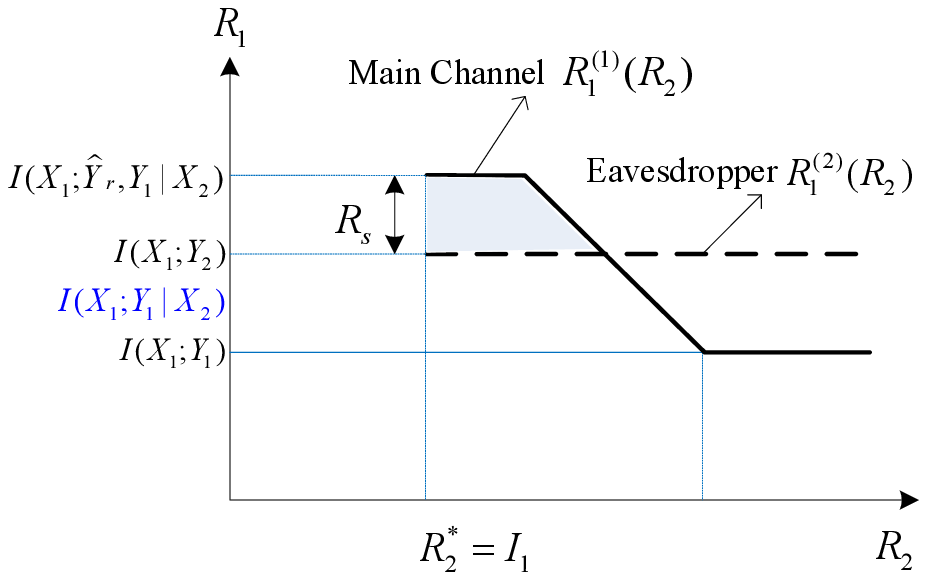}}
\subfigure[Case $ii$ (corresponds to $I_1<I_3$)] {\includegraphics[width=0.4\textwidth]{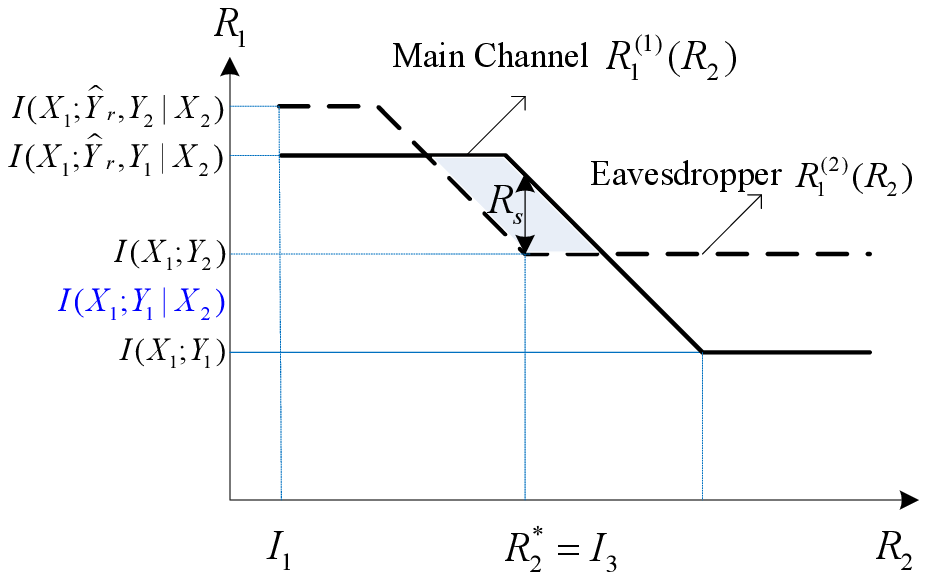}}
\subfigure[Case $iii$ (corresponds to $I_1\geq I_3$)] {\includegraphics[width=0.4\textwidth]{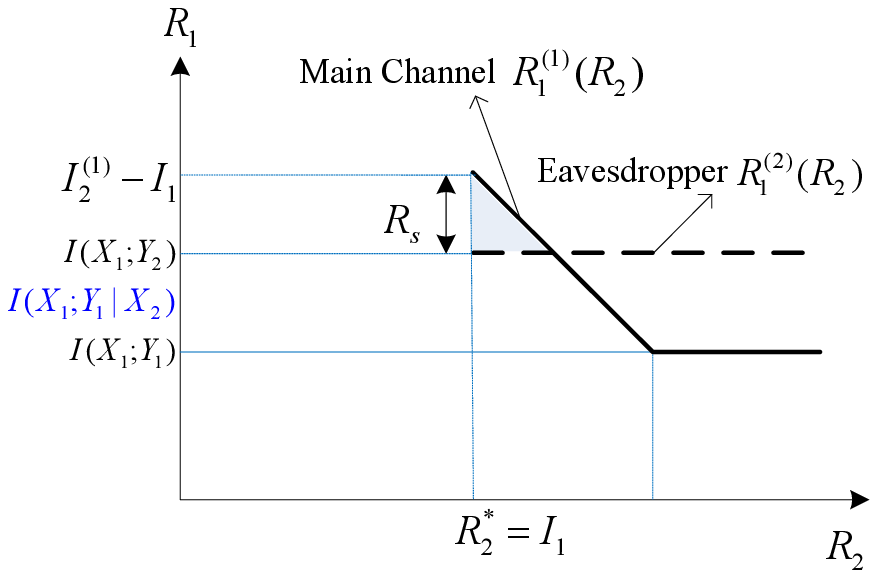}}
\end{center}
  \caption{Illustrations of the functions  $R_1^{(1)}(R_2)$ and $R_1^{(2)}(R_2)$ versus relay rate $R_2(\geq I_1)$ for the very strong eavesdropping case (i.e. $I(X_1;Y_1|X_2)\leq I(X_1;Y_2)$), where $I_3=I(X_2;Y_2|X_1)+I(\hat{Y}_r;X_1,Y_2|X_2)$, $R_2^*=\max\{I_1, I_3\}$. Case $i$, $ii$, $iii$ denotes $R_{s,lb}=I(X_1; \hat{Y}_r,Y_1|X_2) -I(X_1;Y_2)$, $ I_2^{(1)}-I_2^{(2)}$, $I(X_1,X_2;Y_1)-I(\hat{Y}_r;Y_r|X_1,X_2,Y_1)-I(X_1;Y_2)$ in \eqref{very-strong}, respectively. }\label{very_strong}
\end{figure}

\section{Gaussian relay-eavesdropper channel}
In this section, a discrete memoryless Gaussian relay-eavesdropper channel is considered, where the channel outputs at the  destination, eavesdropper, relay can be  expressed as \cite{tang2011interference}
\begin{eqnarray}
  Y_1&=&X_1+\sqrt{b}X_2+Z_1, \nonumber\\
   Y_2&=&\sqrt{a}X_1+X_2+Z_2, \nonumber\\
   Y_r&=&\sqrt{c}X_1+Z_r, \label{gaussian channel}
\end{eqnarray}
where $Z_1,$ $Z_2,$ $Z_r$ are i.i.d. zero-mean Gaussian random variables with unit variance. $a,$ $b,$ $c$ are channels gains. The transmit power of the channel inputs $X_1$ and $X_2$ are constrained by
\begin{eqnarray}
  \frac{1}{n}\sum_{i=1}^n E[X_{1,i}^2] \leq \bar{P}_1,
    \frac{1}{n}\sum_{i=1}^n E[X_{2,i}^2] \leq \bar{P}_2.
\end{eqnarray}
\subsection{Achievable secrecy rate}
Applying Theorem 1 to the Gaussian case given by  \eqref{gaussian channel}, the following theorem can be obtained.
\begin{theorem}\label{thoorem2}
For a Gaussian relay-eavesdropper channel, when fixing the transmit powers at the source and relay as $0\leq P_1 \leq \bar{P}_1$ and $0\leq P_2 \leq \bar{P}_1$, the following secrecy rate is achievable
\begin{eqnarray}
  R_s(P_1,P_2)=\max\{R_s^I(P_1,P_2),R_s^{II}(P_1)\}
  \end{eqnarray}
  where $R_s^I(P_1,P_2)$ is
  \begin{eqnarray}
    R_s^I(P_1,P_2)=\left\{ \begin{array}{ll}
      C\left(P_1+\frac{bcP_1P_2}{1+(1+c)P_1+bP_2}\right)- C\left(\frac{aP_1}{1+P_2}\right),&\textrm{ if }b\geq1+(1+c)P_1\\
      C\left(P_1+bP_2\right)-C\left(aP_1+P_2\right),& \textrm{ if } 1\leq b<1+(1+c)P_1\\
      C\left(\frac{P_1}{1+bP_2}\right)-C\left(\frac{aP_1}{1+P_2} \right), & \textrm{ if } b<1
    \end{array}\right.
  \end{eqnarray}
  and $R_s^{II}(P_1)=[C(P_1)-C(aP_1)]^+$ with $C(x)=\frac{1} {2} \log(1+x)$.
\end{theorem}
\begin{proof}
  refer to Appendix C.
\end{proof}
\begin{Remark}\label{remark power}
  When $b\leq 1+P_1$, the achievable secrecy tare in Theorem 2 is the same as the one in Theorem 3 of \cite{tang2011interference}. However, when $b>1+P_1$, it can be easily proved that the proposed scheme strictly outperforms the latter one. This is mainly because, when  the relay-destination channel is sufficiently strong,  the  main channel becomes a bottleneck for the interference assisted scheme without relaying the source message, whereas the proposed scheme can efficiently improve  the main link by relaying the quantized source message.
\end{Remark}
\begin{Remark}
  For the Gaussian relay-eavesdropper channel, power control can be used to further enhance the achievable secrecy rate in Theorem 2. Under certain conditions, the secrecy rate is not necessarily maximized at $P_1=\bar{P}_1$ and $P_2=\bar{P}_2$. Loosely speaking, in some situations, larger  source power $P_1$ and relay power $P_2$ may enhance the  eavesdropper's decoding capability  more, or  bring more interferences   to the destination. Or in other words,    the maximum power transmission does not always result in  a larger equivocation rate.  Then power control can play an important role for secrecy transmission and change the achievable secrecy rate in Theorem 2 into $$R_s=\max_{0\leq P_1\leq \bar{P}_1,0\leq P_2\leq \bar{P}_2} R_s^I(P_1,P_2),$$ where $R_s^{II}$ is redundant and can be ignored when power control  is applied since it can be viewed as a special case of $R_s^I$ by setting $P_2=0$. Following those steps in Appendix D of \cite{tang2011interference}, the optimal solution of power control can be obtained. Due to the space limit, we will only rely on the computer simulations in this paper to show the performance of the proposed scheme with power control as shown in the  next numerical sections.
\end{Remark}

\subsection{Numerical results}
In Fig.  \ref{compare2}, the achievable secrecy rates for the proposed scheme are shown as a function of the relay-destination channel gain $b$. We set the source-eavesdropper channel gain as $a=1$ or $a=6$ respectively, the capacity of the eavesdropper channel without the relay is zero. The interference assisted scheme in \cite{tang2011interference} and the upper bound obtained in Theorem 1 of \cite{lai2008relay} have also been shown in the figures for comparison.   Interestingly, one can see that the performance of the proposed scheme is the same as the one in \cite{tang2011interference} for $b\leq 1$ and outperforms it when $b>1$. Such a performance gain can be further enlarged    for a large $a$ as $a=6$, which is the very strong case as discussed above. It means that the proposed scheme obtains more helps from the power control compared to the scheme in \cite{tang2011interference}, where the use of power control makes it more likely that the proposed scheme outperforms the comparable one in \cite{tang2011interference}, even with moderate channel conditions. As can be seen from  the Theorem  \ref{thoorem2}, the condition to ensure that the proposed scheme achieves a better performance than the comparable one  is $b>1+P_1$. By applying power control, the value of the source transmission power $P_1$ is not necessarily  to be the maximum $\bar{P}_1$, which brings the benefit that the requirement to the relay-destination channel gain $b$ is reduced.
\begin{figure}[htbp]
\begin{center}
   \subfigure[When $a=1$]{\includegraphics[width=0.45\textwidth]{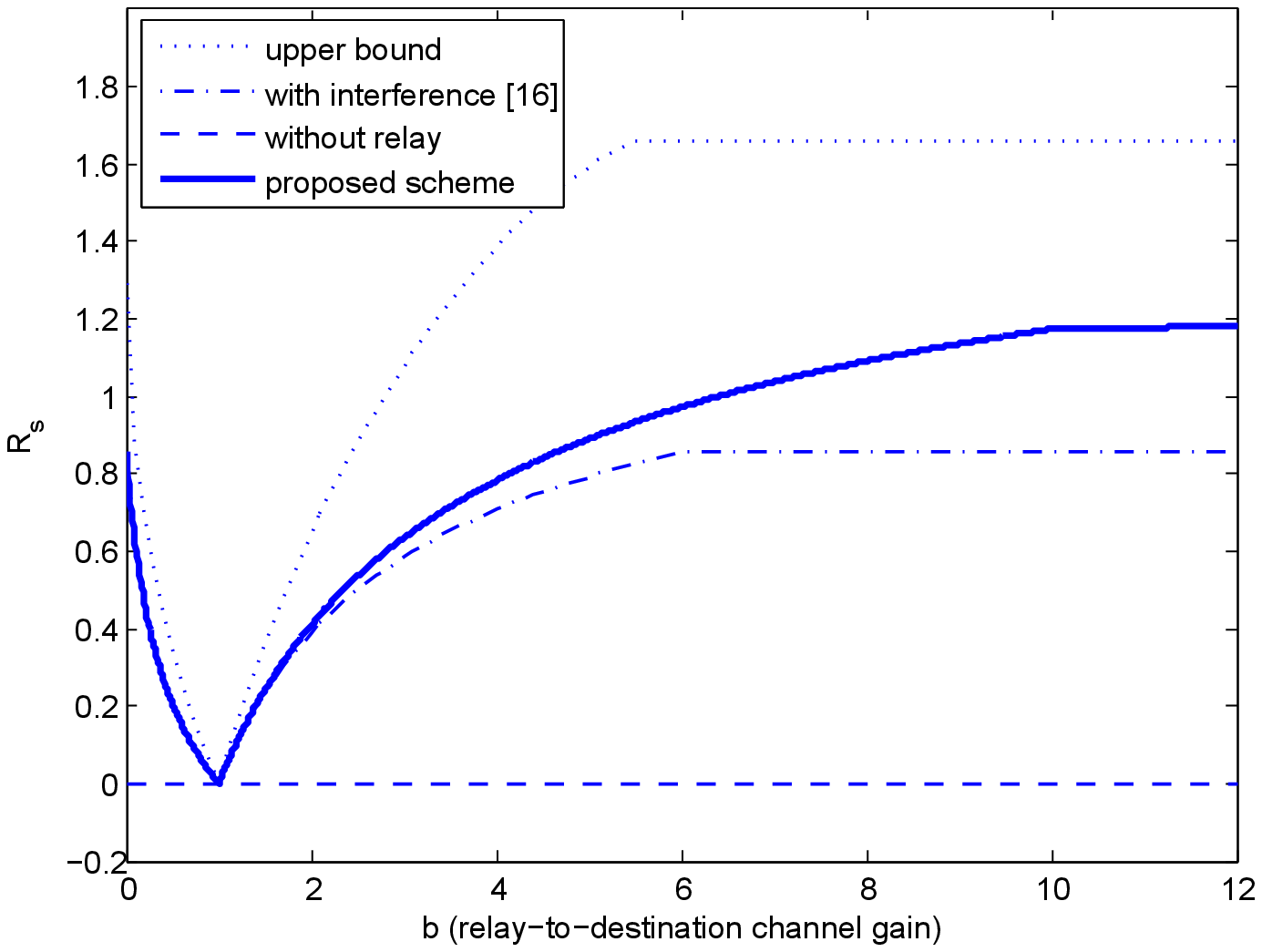}}
   \subfigure[When $a=6$]{\includegraphics[width=0.45\textwidth]{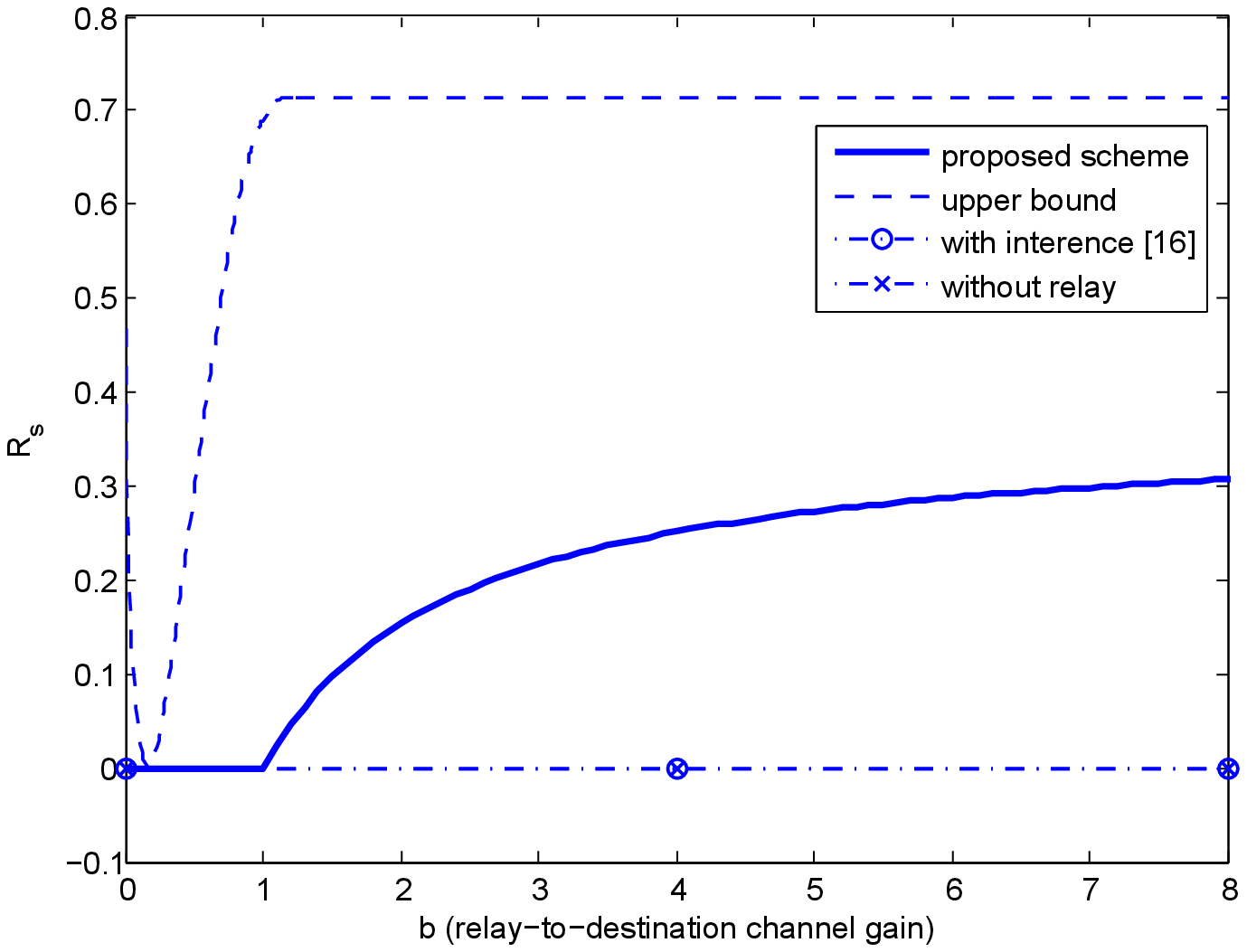}}
\end{center}
\caption{Achievable secrecy rate for different schemes for a relay channel versus $b$ (the relay destination channel gain), where  $c=0.8$ and the power bounds are $\bar{P}_1 =\bar{P}_2=5$.}\label{compare2}
\end{figure}
\section{Conclusion}
  In this paper, we have focused on the four-node relay-eavesdropper channel, and developed a new form of the achievable secrecy rate for such a scenario. The key idea of the   proposed scheme is to combine   noisy network coding   with the  interference assisted strategy for wiretap channel with a helping interferer. Facilitated by these techniques, the  secrecy rate achieved by the proposed scheme was characterized  for both discrete memoryless and  Gaussian channels. Such a new rate can be viewed as a general framework, where the existing interference assisted schemes such as noisy-forwarding  and cooperative jamming approaches can be shown as special cases of the proposed scheme. In addition, under some channel condition where  the existing schemes can only achieve zero secrecy rate, analytic and numerical results have been provided to show that the proposed secrecy scheme can still offer significant performance gains. The impact of power control on the secrecy rate has only been analyzed by relying on computer simulations, where a promising future direction is to carry out the study for the optimal design of power control and allocation for the addressed relay-eavesdropper scenario.

\appendices
\section{Proof of Theorem 1}
{\bf Codebook Generation:} Fix $p(x_1)p(x_2)p(\hat{y}_r|y_r,x_2)$. A codebook for each block is randomly and independently generated. Let $\mathbf{x}_{1j}=(x_{1,(j-1)n+1},\cdots x_{1,jn})$, $j\in[1:B]$. So are the definitions of $\mathbf{x}_{2j},$ $\mathbf{y}_{1j},$ $\mathbf{y}_{2j},$ $\mathbf{y}_{rj},$ $\hat{\mathbf{y}}_{rj}.$

For each $j\in[1:B]$, randomly generate $2^{nBR_1}$ $n$-sequences $\mathbf{x}_{1j}$, each according to the distribution $\prod_{i=1}^n p_{X_1}(x_{1,(j-1)n+i})$. These $2^{nBR_1}$ codewords are then randomly grouped into $2^{nBR_{1,s}}$ bins each with $2^{nBR_{1,o}}$ codewords\footnote{Note that the parameters $R_1$, $R_{1,o}$ correspond to $R_1^{(1)}(R_2)$ and $R_1^{(2)}(R_2)$ in \eqref{R_1(R_2)}, respectively}, $R_1=R_{1,s}+R_{1,o}$. Index them as $\mathbf{x}_{1j}(w_1,w'_1)$, $w_1\in[1:2^{nBR_{1,s}}]$, $w_1'\in[1:2^{nBR_{1,o}}].$ Then randomly generate $2^{nR_2}$ $n$-sequences $\mathbf{x}_{2j}(l_{j-1})$, $l_{j-1}\in[1:2^{nR_2}]$ ($l_0=1$), each according to $\prod_{i=1}^n p_{X_2}(x_{2,(j-1)n+i})$. For each $\mathbf{x}_{2j}(l_{j-1})$, randomly generate $2^{nR_2}$ $n$-sequences $\hat{\mathbf{y}}_{rj}(l_j|l_{j-1})$, $l_{j}\in[1:2^{nR_2}]$, each according to $\prod_{i=1}^n p_{\hat{Y}_r|X_2} (\hat{y}_{r,(j-1)n+i}|x_{2,(j-1)n+i}(l_{j-1}))$.

Hence the codebook is defined as
\begin{eqnarray}
  \mathcal{C}_j=\{\mathbf{x}_{1j}(w_1,w'_1),\mathbf{x}_{2j}(l_{j-1}),
  \hat{\mathbf{y}}_{rj}(l_j|l_{j-1}): \nonumber\\ w_1\in[1:2^{nBR_{1,s}}], w_1'\in[1:2^{nBR_{1,o}}], l_{j}\in[1:2^{nR_2}] \}\nonumber
\end{eqnarray}
for $j\in[1:B]$, ($l_0=1$). Besides, set the rates for $R_1$  and $R_2$ as
\begin{align}\label{R_2}
 & R_2>I_1+\delta(\epsilon'),\\
&R_1<\max\left\{ \min\left[I\left(X_1;\hat{Y}_r,Y_1|X_2\right), I_2^{(1)}-R_2 \right], I(X_1;Y_1) \right\} \nonumber
\end{align}
where $I_1$ and $I_2^{(2)}$ are defined in \eqref{defenition}, $\delta(\epsilon'\rightarrow 0)$ as $\epsilon'\rightarrow 0$.

{\bf Encoding:} The encoding processes at the source and relay are shown as follows:
\begin{itemize}
\item[$\bullet$] To send the confidential  $w_1\in[1:2^{nBR_{1,s}}]$, the source randomly chooses a dummy message $w_1'\in[1:2^{nBR_{1,o}}]$ and then sends $\mathbf{x}_{1j}(w_1,w'_1)$ at block $j$. Meanwhile, the relay transmits $\mathbf{x}_{2j}(l_{j-1})$ in block $j$.
\item[$\bullet$] The relay,
upon the received sequence $\mathbf{y}_{rj}$ at the end of block $j$,  finds an index $l_j$ such that $(\hat{\mathbf{y}}_{rj}(l_j|l_{j-1}),\mathbf{y}_{rj},\mathbf{x}_{2j}(l_{j-1})) \in \mathcal{T}_{\epsilon'}^{(n)}$. If there exists  more than one such {\em qualified} index, uniformly  select one of them at random. Besides, if there doesn't exist such an index, randomly choose an arbitrary index from $[1:2^{nR_2}]$.
\end{itemize}

{\bf Decoding:} Let $\epsilon>\epsilon'$.
 At the end of  block $B$, the destination declares that $\hat{w}_1$ is received if
 \begin{itemize}
 \item[1)] (``joint" decoding of NNC) $\mathbf{x}_{1j}(\hat{w}_1,\hat{w}'_1)$ is the only codeword such that $(\mathbf{x}_{1j}(\hat{w}_1,\hat{w}'_1),\mathbf{x}_{2j}(\hat{l}_{j-1}),
  \hat{\mathbf{y}}_{rj}(\hat{l}_j|\hat{l}_{j-1}), \mathbf{y}_{1j})\in \mathcal{T}_{\epsilon}^{(n)}$ for all $j\in[1:B]$ for some $\hat{l}_1,\cdots \hat{l}_B$; or
  \item[2)] (separate decoding) $\mathbf{x}_{1j}(\hat{w}_1,\hat{w}'_1)$ is the only codeword  such that $(\mathbf{x}_{1j}(\hat{w}_1,\hat{w}'_1), \mathbf{y}_{1j})\in \mathcal{T}_{\epsilon}^{(n)}$ for all $j\in[1:B]$.
 \end{itemize}
The destination makes an error if neither 1) nor 2) occurs, or if there exists more than one such $\hat{w}_1$-index.

For any rate pair  $(R_1,R_2)$ given in \eqref{R_2}, we have $R_1<\min\left[I\left(X_1;\hat{Y}_r,Y_1|X_2\right), I_2^{(1)}-R_2 \right]$ or $R_1<I(X_1;Y_1)$, which is subjected to the constraint in Section III of \cite{lim2011noisy} or the constraint of separate decoding.
Therefore, the intended receiver can decode $W_1$ with an arbitrarily small  probability of error using either the ``joint" decoding strategy according to NNC or the separate decoding way, as long as $n$ and $B$ are sufficiently large.

 \textbf{ Equivocation Computation:}
 The parameter $R_{1,o}$ is different for each of two cases (i.e. Case I and Case II) depending on the code rate $R_2$ at the relay. From \eqref{R_1(R_2)} and the perspective of the eavesdropper, in Case I, let\footnote{Note that if $I_1\geq  I(X_2;Y_2|X_1)+ I(\hat{Y}_r;X_1,Y_2|X_2)$, Case I does not exist, and we only need to consider Case II.} $I_1<R_2< I(X_2;Y_2|X_1)+ I(\hat{Y}_r;X_1,Y_2|X_2)$, which corresponds to the situation $R_1^{(2)}=\min\left[I\left(X_1;\hat{Y}_r,Y_2|X_2\right), I_2^{(2)}-R_2 \right]$; in Case II, let $R_2\geq I(X_2;Y_2|X_1)+ I(\hat{Y}_r;X_1,Y_2|X_2)$, which corresponds to the situation $R_1^{(2)}=I(X_1;Y_2)$.
  Now the equivocation will be lower bounded according to the $R_2$ value as two following  subsections.
 \subsection{Case I: $\left(I_1<R_2\leq I(X_2;Y_2|X_1)+ I(\hat{Y}_r;X_1,Y_2|X_2)\right)$}
  Set the rate parameter for $R_{1,o}$ as
 \begin{eqnarray}\label{R_1,l}
   R_{1,o}=\min \{I(X_1;\hat{Y}_r,Y_2|X_2), I_2^{(2)}-R_2 \}-\epsilon_1.
 \end{eqnarray} In the following analysis, $X_{1j}^n$ denotes $(X_{1,(j-1)n+1},\cdots,X_{1,jn})$, $j\in[1:B]$, then $X_1^{Bn}=(X_{1,1},\cdots,X_{1,Bn})=(X_{11}^n,\cdots,X_{1B}^n)$, and $X_{1,]j[}^n$ denotes the set $\{X_{1k}^n, 1\leq k<j \textrm{ or }j<k\leq n\}$. So are the definitions  with respect to the other variables such as $X_2,$ $Y_2,$ $\hat{Y}_r$.
  The equivocation at the eavesdropper is lower bounded as
 \begin{eqnarray}\label{equivocation}
   H(W_1|Y_2^{Bn}) &\geq & H(W_1|Y_2^{Bn},L_B) \nonumber\\
   &=& H(W_1,Y_2^{Bn}|L_B)-H(Y_2^{Bn}|L_B) \nonumber\\
   &=& H(W_1,X_1^{Bn},X_2^{Bn},\hat{Y}_r^{Bn},Y_2^{Bn}|L_B) \nonumber\\
   && -H(X_1^{Bn},X_2^{Bn},\hat{Y}_r^{Bn}|W_1,Y_2^{Bn},L_B)-H(Y_2^{Bn}|L_B) \nonumber\\
   &=& H(X_1^{Bn},X_2^{Bn},\hat{Y}_r^{Bn}|L_B)+
       H(W_1,Y_2^{Bn}|X_1^{Bn},X_2^{Bn},\hat{Y}_r^{Bn},L_B) \nonumber\\
   && -H(X_1^{Bn},X_2^{Bn},\hat{Y}_r^{Bn}|W_1,Y_2^{Bn},L_B)-H(Y_2^{Bn}|L_B) \nonumber\\
   &\geq& H(X_1^{Bn},X_2^{Bn},\hat{Y}_r^{Bn})-H(L_B)+
   H(Y_2^{Bn}|X_1^{Bn},X_2^{Bn},\hat{Y}_r^{Bn})\nonumber\\
   && -H(X_1^{Bn},X_2^{Bn},\hat{Y}_r^{Bn}|W_1,Y_2^{Bn},L_B)-H(Y_2^{Bn}) \nonumber\\
   &=& H(X_1^{Bn})-H(L_B)+H(X_2^{Bn},\hat{Y}_r^{Bn})-I(X_1^{Bn};X_2^{Bn},\hat{Y}_r^{Bn}) \nonumber\\
   &&-I(X_1^{Bn},X_2^{Bn},\hat{Y}_r^{Bn};Y_2^{Bn})-
   H(X_1^{Bn},X_2^{Bn},\hat{Y}_r^{Bn}|W_1,Y_2^{Bn},L_B)
 \end{eqnarray}
 Now let's calculate the six terms separately in the following   subsections.
 \subsubsection{The first and the second term}
 \begin{eqnarray}\label{term one}
   H(X_1^{Bn})&=&nBR_1=nB(R_{1,s}+R_{1,o})\nonumber\\
   H(L_B)&\leq& \log 2^{nR_2}=nR_2,
 \end{eqnarray}
 %where $H(L_B)$ can be equal to $nR_2$ if and only if $L_B$ is uniformly distributed.
 \subsubsection{The third term} Since a block Markov encoding is used, and $L_1-L_2-\cdots-L_B$ is a Markov chain. So it can be calculated that \begin{eqnarray}\label{H(L^B)}
   H(X_2^{Bn},\hat{Y}_r^{Bn})
    &=&H(L_1,L_2,\cdots L_B) \nonumber\\
   &=&\sum_{j=1}^{B} H(L_j|L_1,\cdots,L_{j-1})\nonumber\\
   &=&\sum_{j=1}^{B} H(L_j|L_{j-1}),
 \end{eqnarray}
 %where $(a)$ follows from the fact that there is a one-to-one mapping between $(X_2^{Bn},\hat{Y}_r^{Bn})$ and $(L_1,L_2,\cdots,L_B)$ when the eavesdropper knows the codebook at each block.
In order to obtain the lower bound of $H(L_B|L_{B-1})$, let's first present a lemma as following.
\begin{Lemma}
For any $j\in[1:B]$ and $l_{j-1},l_j\in[1:2^{nR_2}]$, let $p(l_j|l_{j-1})$ denote $p_{L_j|L_{j-1}}(l_j|l_{j-1})$ for simplicity, the conditional  probability mass function $p(l_j|l_{j-1})$ can be upper bounded as
  \begin{eqnarray}
  p(l_j|l_{j-1})\leq 2^{-n(R_2-2\delta(\epsilon'))} \left(\frac{2}{1-\epsilon'}+\exp\left\{ -\left((1-\epsilon')2^{n(R_2-I_1-\delta(\epsilon'))-3}-n(R_2-2\delta(\epsilon'))\ln2 \right)\right\}\right).
\end{eqnarray}
\end{Lemma}
\begin{proof}
  refer to Appendix B.
\end{proof}
 Therefore, by using the above lemma, $H(L_j|L_{j-1})$ can be bounded as
\begin{eqnarray}
  H(L_j|L_{j-1})&=& \sum_{l_{j-1}=1}^{2^{nR_2}}\sum_{l_{j}=1}^{2^{nR_2}}
  p(l_j,l_{j-1})\log p^{-1}(l_j|l_{j-1})\nonumber\\
  &\geq& \left(\log 2^{n(R_2-2\delta(\epsilon'))}-n\delta_1(n)\right)
  \sum_{l_{j-1}=1}^{2^{nR_2}}\sum_{l_{j}=1}^{2^{nR_2}}p(l_j,l_{j-1})\nonumber\\
  &=&nR_2-n(2\delta(\epsilon')+\delta_1(n)),
\end{eqnarray}
where $\delta_1(n)=\frac{1}{n}\log\left(\frac{2}{1-\epsilon'}+\exp\left\{ -\left((1-\epsilon')2^{n(R_2-I_1-\delta(\epsilon'))-3}-n(R_2-2\delta(\epsilon'))\ln2 \right)\right\}\right) \rightarrow 0$ as $n\rightarrow \infty$ since $R_2$ has been constrained in \eqref{R_2}. Recalling \eqref{H(L^B)}, we can obtain \begin{equation}\label{term two2}
  H(X_2^{Bn},\hat{Y}_r^{Bn})\geq nBR_2-nB(2\delta(\epsilon')+\delta_1(n)).
\end{equation}
\subsubsection{The fourth and fifth terms} The fourth term in \eqref{H(L^B)} can be upper bounded as: \begin{align}
  I(X_1^{Bn};X_2^{Bn},\hat{Y}_r^{Bn}) &= \sum_{j=1}^{B} I(X_1^{Bn};X_{2j}^{n},\hat{Y}_{rj}^{n}|X_{2}^{(j-1)n},\hat{Y}_{r}^{(j-1)n})
  \nonumber\\
  &= \sum_{j=1}^{B} \{ I(X_1^{Bn};\hat{Y}_{rj}^{n}|X_{2j}^{n},X_{2}^{(j-1)n},\hat{Y}_{r}^{(j-1)n})
  +I(X_1^{Bn};X_{2j}^{n}|X_{2}^{(j-1)n},\hat{Y}_{r}^{(j-1)n})\}
  \label{0 term}\\
   &\stackrel{(a)}{\leq} \sum_{j=1}^{B} \{ H(\hat{Y}_{rj}^{n}|X_{2j}^n)- H(\hat{Y}_{rj}^{n}|X_1^{Bn},X_{2j}^{n},X_{2}^{(j-1)n},\hat{Y}_{r}^{(j-1)n})
   \} \nonumber\\
   &\stackrel{(b)}{=} \sum_{j=1}^{B} \{ H(\hat{Y}_{rj}^{n}|X_{2j}^n)- H(\hat{Y}_{rj}^{n}|X_{1j}^{n},X_{2j}^{n})\} \nonumber\\
   &=\sum_{j=1}^B I(X_{1j}^n;\hat{Y}_{rj}^n|X_{2j}^n), \nonumber
\end{align}
where (a) follows removing conditioning and the fact that the second term in
\eqref{0 term} is zero since $X_{2j}^n$ is a deterministic function of $\hat{Y}_{r,j-1}^n, X_{2,j-1}$; (b) follows the fact that $\hat{Y}_{rj}^n -\{X_{1j}^n,X_{2j}^n\}-\{\hat{Y}_{r}^{(j-1)n},X_{2}^{(j-1)n},X_{1,]j[}^n\}$ is a Markov chain, since the channel is memoryless and the codebook for each block is independently generated. Besides, the fifth  term can be upper bounded as
\begin{eqnarray}
  I(X_1^{Bn},X_2^{Bn},\hat{Y}_{r}^{Bn};Y_2^{Bn}) &=& \sum_{j=1}^{B}
    I(X_1^{Bn},X_2^{Bn},\hat{Y}_{r}^{Bn};Y_{2j}^{n}|Y_2^{(j-1)n})\nonumber\\
    &\leq& \sum_{j=1}^B \{H(Y_{2j}^n)-H(Y_{2j}^n|X_1^{Bn},X_2^{Bn}, \hat{Y}_{r}^{Bn},Y_2^{(j-1)n})\} \nonumber\\
    &\stackrel{(a)}{=}& \sum_{j=1}^B \{H(Y_{2j}^n)-H(Y_{2j}^n|X_{1j}^n, X_{2j}^n, \hat{Y}_{rj}^{n}\} \nonumber\\
    &=& \sum_{j=1}^{B}
    I(X_{1j}^{n},X_{2j}^{n},\hat{Y}_{rj}^{n};Y_{2j}^{n}), \nonumber
\end{eqnarray}
where (a) follows from the fact that $Y_{2j}^{n}-\{X_{1j}^{n},X_{2j}^{n}, \hat{Y}_{rj}^{n}\}-\{Y_{2}^{(j-1)n},X_{1]j[}^{n},X_{2]j[}^{n}, \hat{Y}_{r]j[}^{n}\}$ is a Markov chain. Using the above two results, we can bound the sum of above two terms as
\begin{eqnarray}\label{sum 4,5}
  S_{4,5}&=& I(X_1^{Bn};X_2^{Bn},\hat{Y}_r^{Bn})+I(X_1^{Bn},X_2^{Bn}, \hat{Y}_{r}^{Bn};Y_2^{Bn})\nonumber\\
  &\leq& \sum_{j=1}^B [I(X_{1j}^n;\hat{Y}_{rj}^n|X_{2j}^n)+I(X_{1j}^{n}, X_{2j}^{n},\hat{Y}_{rj}^{n};Y_{2j}^{n})] \nonumber\\
  &=& \sum_{j=1}^B [I(X_{1j}^{n}, X_{2j}^{n};Y_{2j}^{n})+I(\hat{Y}_{rj}^n;X_{1j}^n,Y_{2j}^n|X_{2j}^n)] \nonumber\\
  &\leq& \sum_{j=1}^B n[I(X_1,X_2;Y_2)+I(\hat{Y}_r;X_1,Y_2|X_2)+\delta_2(n)]. \nonumber
\end{eqnarray}
On the other hand,
\begin{eqnarray}
  S_{4,5}&\leq& \sum_{j=1}^B [I(X_{1j}^n;\hat{Y}_{rj}^n|X_{2j}^n)+I(X_{1j}^{n}, X_{2j}^{n},\hat{Y}_{rj}^{n};Y_{2j}^{n})] \nonumber\\
  &=& \sum_{j=1}^B [I(X_{1j}^n;\hat{Y}_{rj}^n,Y_{2j}^n|X_{2j}^n)+I( X_{2j}^{n},\hat{Y}_{rj}^{n};Y_{2j}^{n})] \nonumber\\
  &\leq& \sum_{j=1}^B [I(X_{1j}^n;\hat{Y}_{rj}^n,Y_{2j}^n|X_{2j}^n) +H(X_{2j}^n)] \nonumber\\
  &\leq& \sum_{j=1}^B n[I(X_1;\hat{Y}_r,Y_2|X_2)+R_2+\delta_2(n)]. \nonumber
\end{eqnarray}
Therefore, we have
\begin{eqnarray}\label{term three}
  S_{4,5}\leq nB(R_{1,o}+R_2)+nB(\epsilon_1+\delta_2(n)),
\end{eqnarray}
where $R_{1,o}$ is defined in \eqref{R_1,l}.
\subsubsection{The last term}
Now let's bound the last term $H(X_1^{Bn},X_2^{Bn},\hat{Y}_r^{Bn}|W_1,Y_2^{Bn},L_B)$. Fix $W_1=w_1$ and $L_B=l_b$, and assume that the source sends a codeword $\mathbf{x}_{1j}(w_1,w'_1)$ and the relay transmits $\mathbf{x}_{2j}(l_{j-1})$ at each block $j\in[1:B]$.  Since $X_1^{Bn}$ is determined by two indices $w_1$, $w'_1$ and $(X_2^{Bn},\hat{Y}_r^{Bn})$ are determined by $B$ indices $(l_1,\cdots,l_B)$,   the eavesdropper only needs to do   joint decoding of $w'_1$ and $(l_1,\cdots,l_{B-1})$ at the end of block $B$ with the side information that $W_1=w_1$ and  $L_B=l_b$. Specifically, for any $\epsilon>\epsilon'$, given $W_1=w_1$ and $L_B=l_B$, the eavesdropper finds the unique index set  $\{\hat{w}'_1,\hat{l}_1, \cdots,\hat{l}_{B-1}\}$ such that $(\mathbf{x}_{1j}({w}_1,\hat{w}'_1),\mathbf{x}_{2j}(\hat{l}_{j-1}),
  \hat{\mathbf{y}}_{rj}(\hat{l}_j|\hat{l}_{j-1}), \mathbf{y}_{2j})\in \mathcal{T}_{\epsilon}^{(n)}$ for all $j\in[1:B]$.  The eavesdropper make an error if there exists none or  more than one such index set.

 Analysis of the error probability: Given $W_1=w_1$ and $L_B=l_B$, assume without loss of generality that $W'_1=1$ and $L_1=\cdots=L_B=1$ are sent. Then the eavesdropper makes an error only if at least one of the following events occur:
\begin{eqnarray}
  \varepsilon_1&=& \{(\hat{\mathbf{Y}}_{rj}({l}_j|1), \mathbf{X}_{2j}(1), \mathbf{Y}_{rj}) \not\in \mathcal{T}_{\epsilon'}^{(n)} \textrm{ for all } l_j\in[1:2^{nR_2}] \textrm{ for some } j\in[1:B] \} \nonumber\\
  \varepsilon_2&=& \{(\mathbf{X}_{1j}(w_1,1), \mathbf{X}_{2j}(1),\hat{\mathbf{Y}}_{rj}(1|1),  \mathbf{Y}_{2j}) \not\in \mathcal{T}_{\epsilon}^{(n)} \textrm{ for some } j\in[1:B] \} \nonumber\\
  \varepsilon_3&=& \{(\mathbf{X}_{1j}(w_1,w'_1), \mathbf{X}_{2j}(l_{j-1}),\hat{\mathbf{Y}}_{rj}({l}_j|l_{j-1}),  \mathbf{Y}_{2j}) \in \mathcal{T}_{\epsilon}^{(n)} \textrm{ for all } j \textrm{ for some } l^{B-1}, w'_1\neq 1 \} \nonumber\\
  \varepsilon_4&=& \{(\mathbf{X}_{1j}(w_1,1), \mathbf{X}_{2j}(l_{j-1}),\hat{\mathbf{Y}}_{rj}({l}_j|l_{j-1}),  \mathbf{Y}_{2j}) \in \mathcal{T}_{\epsilon}^{(n)}
  \textrm{ for all } j \textrm{ for some } l^{B-1}\neq 1^{B-1} \} \nonumber
\end{eqnarray}
Thus the error probability can be bounded as
\begin{eqnarray}
  P(\varepsilon) &=& P(\varepsilon\cap \varepsilon_1) +P(\varepsilon \cap \varepsilon_1^c) \nonumber\\
  &=& P(\varepsilon\cap\varepsilon_1)+P\left((\varepsilon_2 \cup \varepsilon_3 \cup \varepsilon_4)\cap \varepsilon_1^c\right) \nonumber\\
  &=&P(\varepsilon\cap\varepsilon_1)+P\left((\varepsilon_2\cap \varepsilon_1^c) \cup (\varepsilon_3\cap \varepsilon_1^c) \cup (\varepsilon_4\cap \varepsilon_1^c)\right) \nonumber\\
  &\leq& P(\varepsilon_1) + P(\varepsilon_2\cap \varepsilon_1^c) +P(\varepsilon_3\cap \varepsilon_1^c) +P(\varepsilon_4\cap \varepsilon_1^c) \nonumber\\
  &\leq& P(\varepsilon_1)+ P(\varepsilon_2 {\small \cap } \varepsilon_1^c) + P(\varepsilon_3)+ P(\varepsilon_4).
\end{eqnarray}
From \cite{lim2011noisy}, one can see that the first three terms goes to 0 as $n\rightarrow \infty $ since we have constrained the rate pair $(R_2,R_{1,o})$ in \eqref{R_2} and \eqref{R_1,l}. For $ P(\varepsilon_4)$, define the events
\begin{eqnarray}
\tilde{\varepsilon}_j(1,l_{j-1},l_j)=\{(\mathbf{X}_{1j}(w_1,1), \mathbf{X}_{2j}(l_{j-1}),\hat{\mathbf{Y}}_{rj}({l}_j|l_{j-1}),  \mathbf{Y}_{2j}) \in \mathcal{T}_{\epsilon}^{(n)}\}, \nonumber \end{eqnarray}
and then
\begin{eqnarray}\label{p(e4)}
  P(\varepsilon_4) &=& P(\cup_{l^{B-1}\neq 1^{B-1}} \cap_{j=1}^B \tilde {\varepsilon}_j(1,l_{j-1},l_j)) \nonumber\\
  &\leq& \sum_{l^{B-1}\neq 1^{B-1}} P( \cap_{j=1}^B \tilde {\varepsilon}_j(1,l_{j-1},l_j)) \nonumber\\
  &\stackrel{(a)}{=}& \sum_{l^{B-1}\neq 1^{B-1}} \prod_{j=1}^B P(  \tilde {\varepsilon}_j(1,l_{j-1},l_j)) \nonumber\\
  &\leq& \sum_{l^{B-1}\neq 1^{B-1}} \prod_{j=2}^B P(  \tilde {\varepsilon}_j(1,l_{j-1},l_j)),
\end{eqnarray}
where $(a)$ is due to the fact that the codebook at each block $j\in[1:b]$ is independently generated and the  memoryless channel is considered.
Note that if $w'_1=1$ and $l_{j-1}\neq 1$,
$(\mathbf{X}_{2j}(l_{j-1}), \hat{\mathbf{Y}}_{rj}({l}_j|l_{j-1}))\sim \prod_{i=1}^n p_{X_2,\hat{Y}_r}(x_{2,(j-1)n+i},\hat{y}_{r,(j-1)n+i})
$ is independent of $(\mathbf{X}_{1j}(w_1,1),\mathbf{Y}_{2j})$. So by the joint typicality
lemma (\cite{gamal2010lecture}, Lecture Note 2), for $w'_1=1$ and $l_{j-1}\neq 1$,
%\begin{eqnarray}
%  P(\tilde {\varepsilon}_j(1,l_{j-1},l_j)) \leq 2^{-n(I_3-\delta(\epsilon))}
%\end{eqnarray}
\begin{eqnarray}
  P(\tilde {\varepsilon}_j(1,l_{j-1},l_j))&=& P\left((\mathbf{X}_{1j}(w_1,1), \mathbf{X}_{2j}(l_{j-1}),\hat{\mathbf{Y}}_{rj}({l}_j|l_{j-1}),  \mathbf{Y}_{2j}) \in \mathcal{T}_{\epsilon}^{(n)}\right) \nonumber\\
  &=&\sum_{(\mathbf{x}_{1j},\mathbf{x}_{2j},\hat{\mathbf{y}}_{rj},\mathbf{y}_{2j}) \in \mathcal{T}_{\epsilon}^{(n)} } p(\mathbf{x}_{2j},\hat{\mathbf{y}}_{rj}) p(\mathbf{x}_{1j},\mathbf{y}_{2j})\nonumber\\
  &\leq& |\mathcal{T}_{\epsilon}^{(n)}|2^{-n(H(X_2,\hat{Y}_r) -\epsilon H(X_2,\hat{Y}_r) )} 2^{-n(H(X_1,Y_2) -\epsilon H(X_1,Y_2))} \nonumber\\
  &\leq& 2^{-n(H(X_1,Y_2)+H(X_2,\hat{Y}_r) -H(X_1,X_2,\hat{Y}_r,Y_2) -\delta (\epsilon))} \nonumber\\
   &=& 2^{-n(I_3-\delta(\epsilon))},
\end{eqnarray}
where $I_3=I(X_2,\hat{Y}_r;X_1,Y_2)=I(X_2;Y_2|X_1)+I(\hat{Y}_r;X_1Y_2|X_2)$. If conditioned on that the binary sequence $l^{B-1}$ has $k(k\in[1:B-2])$ 1s, from the above result we have \begin{eqnarray}
 \prod_{j=2}^B P(  \tilde {\varepsilon}_j(1,l_{j-1},l_j))
 \leq 2^{-n(B-1-k)(I_3-\delta(\epsilon))}.
\end{eqnarray}
Hence substitute the above result   to \eqref{p(e4)} and obtain
\begin{eqnarray}
 P(\varepsilon_4) &\leq& \sum_{l^{B-1}\neq 1^{B-1}} \prod_{j=2}^B P(  \tilde {\varepsilon}_j(1,l_{j-1},l_j)) \nonumber\\
 &\leq& \sum_{k=0}^{B-2} \left( \begin{matrix}
 B-1\\ k \end{matrix} \right) 2^{n(B-1-k)R_2}2^{-n(B-1-k)(I_3-\delta(\epsilon))} \nonumber\\
 &=&\sum_{i=1}^{B-1} \left( \begin{matrix}
 B-1 \\ i \end{matrix} \right) 2^{-nia_0}, \nonumber
\end{eqnarray}
where $i=B-1-k$ and $a_0=I_3-R_2-\delta(\epsilon)$. Note that $a_0>0$ since Case I is considered here. Let  $u_i=\left( \begin{matrix}
 B-1 \\ k \end{matrix} \right) 2^{-nia_0}$, then $$ \frac{u_{i+1}}{u_i}=\frac{B-1-i}{i+1} 2^{-na_0}\leq (B/2-1)2^{-na_0}, $$ for $ 1\leq i\leq B-2$. Hence $u_i\leq u_1((B/2-1)2^{-na_0})^{i-1}$ for $1\leq i\leq B-1$, where $u_1=(B-1) 2^{-na_0}$. Therefore $$ P(\epsilon _4)\leq (B-1)2^{-na_0}\sum_{i=1}^{B-1} ((B/2-1)2^{-na_0})^{i-1}\leq \frac{(B-1)2^{-na_0}} {1-(B/2-1)2^{-na_0}}.$$
Therefore, the error probability $P(\epsilon)$ goes to 0 if $n$ is sufficiently large. From Fano's inequality, we have \begin{eqnarray}&&\frac{1}{nB}H(X_1^{Bn},X_2^{Bn},\hat{Y}_r^{Bn}|W_1=w_1,L_B=l_B, Y_2^{Bn})\nonumber\\
&& \leq \frac{1}{nB}\left( 1+P(\epsilon)\log (|\mathcal{W}'_1|\times |\mathcal{L}^{B-1}|) \right) \nonumber\\
&&\leq\frac{1}{nB}\left( 1+P(\epsilon)(\log 2^{nBR_{1,o}}+(B-1)\log 2^{nR_2})  \right) \nonumber\\
&& =\frac{1}{nB}+P(\epsilon) \left(R_{1,o}+\frac{B-1}{B}R_2\right)\nonumber\\
&&= \delta_3(n).
\end{eqnarray}
 Thus the last term can be bounded as \begin{eqnarray}\label{term four}
 &&\frac{1}{nB} H(X_1^{Bn},X_2^{Bn},\hat{Y}_r^{Bn}|W_1,Y_2^{Bn},L_B)\nonumber\\
  && =\sum_{w_1,l_B} p(w_1,l_B)H(X_1^{Bn},X_2^{Bn},\hat{Y}_r^{Bn}|W_1=w_1,L_B=l_B,Y_2^{Bn})
  \nonumber\\
  && \leq \delta_3(n).
\end{eqnarray}
By combining \eqref{equivocation} with \eqref{term one}, \eqref{term two2}, \eqref{term three} and \eqref{term four}, we get \begin{eqnarray}
\frac{1}{nB}H(W_1|Y_2^{Bn})\geq R_{1,s}-\frac{R_2}{B}-2\delta(\epsilon')-\delta_1(n)-\epsilon_1-\delta_2(n)- \delta_3(n).
\end{eqnarray}
By letting $B\rightarrow \infty$, the equivocation is approaching to the secrecy rate.
\subsection{Case II: $\left(R_2\geq I(X_2;Y_2|X_1)+ I(\hat{Y}_r;X_1,Y_2|X_2)\right)$}
This proof can be completed  by following the similar steps in Case I and \cite{tang2011interference}. First  the rate parameter for $R_{1,o}$ is chosen as \begin{eqnarray}\label{II-R_{1,o}}
  R_{1,o}= I(X_1;Y_2)-\epsilon_1.
\end{eqnarray}

To prove the rate  for CASE II, the relay uses a similar binning procedure as the the source. For each $j\in[1:B]$, the $2^{nR_2}$ codewords are randomly grouped into $2^{nR'_2}$ bins each with $2^{nR''_2}$ codewords, hence $R_2=R'_2+R''_2$. So that the index of each codeword can be equivalently expressed as $l_j=(l'_j,l''_j)$, where $l'_j\in[1:2^{nR'_2}],$ $l''_j\in[1:2^{nR''_2}]$. To simplify  the proof, we set
\begin{eqnarray}\label{II-R''_2}
  R''_2=I(X_2;Y_2|X_1)+ I(\hat{Y}_r;X_1,Y_2|X_2)-\epsilon_2.
\end{eqnarray}
Then following the steps in \eqref{equivocation} and let $(L')^{B-1} =\{ L'_1,\cdots, L'_{B-1}\}$, the equivocation can be bounded as
\begin{eqnarray}\label{II-equivocation}
   H(W_1|Y_2^{Bn}) &\geq & H(W_1|Y_2^{Bn},L_B,(L')^{B-1}) \\
   &\geq& H(X_1^{Bn})-H(L_B,(L')^{B-1})+H(X_2^{Bn},\hat{Y}_r^{Bn})
   -I(X_1^{Bn};X_2^{Bn},\hat{Y}_r^{Bn}) \nonumber\\
   &&-I(X_1^{Bn},X_2^{Bn},\hat{Y}_r^{Bn};Y_2^{Bn})-
   H(X_1^{Bn},X_2^{Bn},\hat{Y}_r^{Bn}|W_1,Y_2^{Bn},L_B,(L')^{B-1}), \nonumber
 \end{eqnarray}
 where $H(X_1^{Bn})=nB(R_{1,s}+R_{1,o})$, $H(X_2^{Bn},\hat{Y}_r^{Bn})\geq nBR_2-nB(2\delta(\epsilon')+\delta_1(n))$ as shown in \eqref{term two2}. The second term can be bounded as $$ H(L_B,(L')^{B-1})\leq H(L_B)+\sum_{j=1}^{B-1} H(L'_j)\leq nR_2+n(B-1)R'_2. $$ Besides, from \eqref{sum 4,5}, the sum of the fourth and the fifth terms is
 \begin{eqnarray}
  S_{4,5}&=& I(X_1^{Bn};X_2^{Bn},\hat{Y}_r^{Bn})+I(X_1^{Bn},X_2^{Bn}, \hat{Y}_{r}^{Bn};Y_2^{Bn})\nonumber\\
  &\leq&  nB[I(X_1,X_2;Y_2)+I(\hat{Y}_r;X_1,Y_2|X_2)+\delta_2(n)]. \nonumber\\
  &=&  nB[R_{1,o}+R''_2+\epsilon_1+\epsilon_2+\delta_2(n)].
\end{eqnarray}

To bound the last term $H(X_1^{Bn},X_2^{Bn},\hat{Y}_r^{Bn}|W_1,Y_2^{Bn},L_B,(L')^{B-1})$, the eavesdropper only needs to do   joint decoding of $W'_1$ and $\{L''_1,\cdots,L''_B\}$ at the end of block $B$ assuming that $W_1$,  $L_B$ and $(L')^{B-1}$ are given to it as side information. Similar to the analysis of the last term in Case I, for the rates $(R_2,R_{1,o}, R''_2)$ constrained in \eqref{R_2}, \eqref{II-R_{1,o}} and \eqref{II-R''_2}, it can be shown that the error probability is arbitrarily small for sufficiently large $n$. Hence we have
 \begin{eqnarray}
\frac{1}{nB} H(X_1^{Bn},X_2^{Bn},\hat{Y}_r^{Bn}|W_1,Y_2^{Bn},L_B,(L')^{B-1})
\leq \delta_3(n). \nonumber
\end{eqnarray}

Substituting the above results into \eqref{II-equivocation}, the equivocation can be bounded as
\begin{eqnarray}
  \frac{1}{nB}H(W_1|Y_2^{Bn})\geq R_{1,s}-\frac{R''_2}{B}-2\delta(\epsilon')-\delta_1(n)-\epsilon_1-\epsilon
  _2-\delta_2(n)- \delta_3(n).
\end{eqnarray}
Again by letting $B\rightarrow \infty $, we can show that the equivocation is approaching the secrecy rate.

\section{Proof of Lemma 1}
To upper bound the conditional probability mass function $p(l_j|l_{j-1})$ for any $j\in[1:B]$ and $l_{j-1},l_j \in [1:2^{nR_2}]$, we first make  some useful definitions as follows.
\begin{itemize}
  \item[$\bullet$] Given $L_{j-1}=l_{j-1}$, define $2^{nR_2}$ binary random variables as
  \begin{eqnarray}
    Q_{k,l_{j-1}}=\left\{ \begin{array}{ll}
      1, & \textrm{ if } \{(\hat{\mathbf{Y}}_{rj}(k|l_{j-1}),\mathbf{Y}_{rj},\mathbf{X}_{2j}(l_{j-1}))\in
  \mathcal{T}_{\epsilon'}^{(n)}|L_{j-1}=l_{j-1}\}\\
  0, & \textrm{ if } \{(\hat{\mathbf{Y}}_{rj}(k|l_{j-1}),\mathbf{Y}_{rj},\mathbf{X}_{2j}(l_{j-1}))\not\in
  \mathcal{T}_{\epsilon'}^{(n)}|L_{j-1}=l_{j-1}\}
    \end{array} \right., \nonumber
  \end{eqnarray}
where $k\in[1:2^{nR_2}]$. Define a new random variable as $T_{l_{j-1}}=\sum_{k=1}^{2^{nR_2}} Q_{k,l_{j-1}}$, which represents that there are $T_{l_{j-1}}$ qualified indexes in the $2^{nR_2}$ $\hat{\mathbf{y}}_{rj}$-codewords.
From the joint typicality (\cite{gamal2010lecture}, Lecture Note 2), for sufficiently large $n$ the probability $P(Q_{k,l_{j-1}}=1)$ can be bounded as:
\begin{eqnarray}\label{P(Q_k)}
  (1-\epsilon')2^{-n(I_1+\delta(\epsilon'))}\leq P(Q_{k,l_{j-1}}=1) \leq 2^{-n(I_1-\delta(\epsilon'))},
\end{eqnarray}
where $I_1=I(\hat{Y}_r;Y_r|X_2).$ The expectation of $T_{l_{j-1}}$ can be expressed as $E(T_{l_{j-1}})=\sum_{k=1}^{2^{nR_2}} P(Q_{k,l_{j-1}}=1)$, so it can be bounded as
\begin{eqnarray}\label{E(Q)}
  (1-\epsilon')2^{n(R_2-I_1-\delta(\epsilon'))}\leq E(T_{l_{j-1}}) \leq 2^{n(R_2-I_1+\delta(\epsilon'))}.
\end{eqnarray}
\item[$\bullet$] Let  $L_{j,l_{j-1}}\sim p_{L_j|L_{j-1}}(l_j|l_{j-1})$, {\em i.e.} $L_{j,l_{j-1}}\sim p(l_j|l_{j-1})$.
Similar to the definitions in \cite{gamal2010lecture} (Lecture Note 4), we can identify the random variables  $L_j:= L_{j,L_{j-1}}$, $Q_k:= Q_{k,L_{j-1}}$ and $T:= T_{L_{j-1}}$, whose distributions depend on $L_{j-1}$ in the same way as the distributions of $L_{j,l_{j-1}}$, $Q_{k,l_{j-1}}$ and $T_{l_{j-1}}$ depend on $l_{j-1}$.  From the encoding process, given $L_{j-1}=l_{j-1}$, conditioned on that there are $t(t>0)$ qualified indexes and $l_{j}$ is one of them ({\em i.e.} $Q_{l_j,l_{j-1}}=1$), obviously we have   \begin{eqnarray}\label{P(u|)}
      P\{ {L}_{j,l_{j-1}}=l_j|Q_{l_j,l_{j-1}}=1,T_{l_{j-1}}=t\}=
      \frac{1}{t}.
    \end{eqnarray}
 Similarly, given $L_{j-1}=l_{j-1}$, we can obtain another conditional probability as \begin{eqnarray}\label{P(u|2)}
  P\{ {L}_{j,l_{j-1}}=l_j|Q_{l_j,l_{j-1}}=0,T_{l_{j-1}}=t\}=0, \textrm{ for } t>0.
\end{eqnarray}
\end{itemize}
Let  $M=\frac{1}{2} E(T_{l_{j-1}})$ which is assumed to be an integer without loss of generality, then the probability $p(l_j|l_{j-1})$ can be calculated as
\begin{align}\label{p(u)}
  p(l_j|l_{j-1})&=P(L_{j,{L_{j-1}}}=l_j|{L_{j-1}}={l_{j-1}})=P( {L}_{j,{l_{j-1}}}=l_j)\nonumber\\
  &=\sum_{t=0}^{2^{nR_2}} P(\{ {L}_{j,{l_{j-1}}}=l_j\} \cap \{T_{l_{j-1}}=t\}) \nonumber\\ &= \sum_{t=0}^{M-1} P(\{ {L}_{j,{l_{j-1}}}=l_j\} \cap \{T_{l_{j-1}}=t\}) +\sum_{t=M}^{2^{nR_2}} P(\{ {L}_{j,{l_{j-1}}}=l_j\} \cap \{T_{l_{j-1}}=t\}).
  \end{align}
  The first term in the above equation can be bounded as
\begin{align}
\sum_{t=0}^{M-1} P(\{ {L}_{j,{l_{j-1}}}=l_j\} & \cap \{T_{l_{j-1}}=t\}) \leq \sum_{t=0}^{M-1} P(T_{l_{j-1}}=t) \nonumber\\&= P(T_{l_{j-1}}<M)\nonumber\\
&\stackrel{(a)}{\leq} \exp \left\{ -\frac{E(T_{l_{j-1}})}{8} \right\} \nonumber\\ &\leq \exp \{ - (1-\epsilon')2^{n(R_2-I_1-\delta(\epsilon'))-3} \},
  \end{align}
  where $(a)$ is based on the {\em multiplicative form of Chernoff bound} ( Eq. (7) in \cite{hagerup1990guided}) by setting the \emph{relative} error as $\frac{1}{2}$. Besides, the second term can be upper bounded as
  \begin{align}
  \sum_{t=M}^{2^{nR_2}} P(\{ {L}_{j,{l_{j-1}}}=l_j\} &\cap \{T_{l_{j-1}}=t\})=\sum_{t=M}^{2^{nR_2}} \sum_{i=0}^1 P(\{ {L}_{j,{l_{j-1}}}=l_j\} \cap \{T_{l_{j-1}}=t\}\cap \{Q_{l_j,{l_{j-1}}}=i\}) \nonumber\\
  &=\sum_{t=M}^{2^{nR_2}} \sum_{i=0}^1  P(Q_{l_j,{l_{j-1}}}=i,T_{l_{j-1}}=t) P( {L}_{j,{l_{j-1}}}=l_j | Q_{l_j,{l_{j-1}}}=i,T_{l_{j-1}}=t)  \nonumber\\
  &\stackrel{(a)}{=} \sum_{t=M}^{2^{nR_2}} \frac{1}{t} P(Q_{l_j,{l_{j-1}}}=1,T_{l_{j-1}}=t)
  \nonumber\\ &\leq \frac{2}{E(T_{l_{j-1}})}P(Q_{l_j,{l_{j-1}}}=1) \nonumber\\
  &\stackrel{(b)}{\leq} \frac{1}{(1-\epsilon')}2^{-n(R_2-2\delta(\epsilon'))+1},
\end{align}
where $(a)$ is based on  \eqref{P(u|)} and \eqref{P(u|2)}, $(b)$ is based on  \eqref{P(Q_k)} and \eqref{E(Q)}. Substituting  the above two results into \eqref{p(u)} concludes the proof of Lemma 1.

\section{Proof of Theorem 2}
Based on the  coding scheme  in Appendix A, $\mathcal{C}_1$ and $\mathcal{C}_2$ are generated by letting the distributions $\pi$ as $X_1\sim \mathcal{N}(0,P_1)$,  $X_2\sim \mathcal{N}(0,P_2)$, $\hat{Y}_r=Y_r+Z_C$ where $Z_C\sim \mathcal{N}(0,\delta_C)$ and $Z_C$ is independent of any other variable. Assuming  the fixed power pair ($P_1,P_2$) is used, here we only discuss the coding parameters $R_2$ and $\delta_C$. From \eqref{R_1(R_2)}, $R_1^{(1)}(R_2,\delta_C)$ and $R_1^{(2)}(R_2,\delta_C)$ can be calculated as
\begin{align}
 R_1^{(1)}(R_2,\delta_C)=\max\left\{\left[C\left(P_1+\frac{cP_1}{1
 +\delta_C}\right), C(P_1+bP_2)+C\left( \frac{cP_1}{1+\delta_C}\right)-R_2\right],C\left( \frac{P_1}{1+bP_2}\right) \right\} \nonumber\\
 R_1^{(2)}(R_2,\delta_C)=\max\left\{\left[C\left(aP_1+\frac{cP_1}{1
 +\delta_C}\right), C(aP_1+P_2)+C\left( \frac{cP_1}{1+\delta_C}\right)-R_2\right],C\left( \frac{aP_1}{1+P_2}\right) \right\}. \nonumber
\end{align}

 Next, we can observe that one achievable choice of  $R_S^I$ is to set $\delta_C=\delta_C^*$ with $\delta_C^* =\frac{1+(1+c)P_1}{bP_2}$  and $R_2=R_2^*$ with $R_2^*=\max\left\{ C\left( \frac{1+cP_1}{\delta_C^*}\right),C(P_2)+C\left( \frac{cP_1}{1+\delta_C^*}\right) \right\}$. In this case, $$R_1^{(1)}(R_2^*,\delta_C^*)=\max\left\{ C(P_1+bP_2)+C\left( \frac{cP_1}{1+\delta_C^*}\right)-R_2^*,C\left( \frac{P_1}{1+bP_2}\right) \right\}$$ and $R_1^{(2)}(R_2^*,\delta_C^*)= C\left( \frac{aP_1}{1+P_2}\right)$ can be  obtained.
 \subsubsection{When $b\geq 1+(1+c)P_1$} $\delta_C^*\leq \frac{1}{P_2}$, we have $R_2^*=C\left( \frac{1+cP_1}{\delta_C^*}\right)$. Then we calculate   $R_1^{(1)}=C\left(P_1+\frac{cP_1}{1+\delta_C^*}\right)$. And the secrecy rate $R_s=R_1^{(1)}-R_1^{(2)}=C\left(P_1+\frac{bcP_1P_2}{1+(1+c)P_1+bP_2}\right)- C\left(\frac{aP_1}{1+P_2}\right)$ can be obtained.
 \subsubsection{When $1\leq b< 1+(1+c)P_1$} $\frac{1}{P_2} \leq \delta_C^*<\frac{1+(1+c)P_1}{P_2}$, we have $R_2^*=C(P_2)+C\left(\frac{cP_1}{1+\delta_C^*}\right)$. Then $R_1^{(1)}$ can be calculated as $R_1^{(1)}=C(P_1+bP_2)-C(P_2)$. So the secrecy rate is $R_s= R_1^{(1)}-R_1^{(2)}=C(P_1+bP_2)-C(aP_1+P_2)$.
 \subsubsection{When $b\leq 1$} $\delta_C^*\geq\frac{1+(1+c)P_1}{P_2}$, we have $R_2^*=C(P_2)+C\left(\frac{cP_1}{1+\delta_C^*}\right)$. Then $R_1^{(1)}$ can be calculated as $R_1^{(1)}=C\left(\frac{P_1}{1+bP_2}\right)$. So the secrecy rate is $R_s=C\left(\frac{P_1}{1+bP_2}\right)-C\left( \frac{aP_1}{1+P_2}\right)$.

 On the other hand, $R_s^{II}=[C(P_1)-C(aP_1)]^+$ can be achieved by setting $\delta_C=\infty$ and $R_2=0$, which may yield a higher secrecy rate under certain conditions. Hence $\max\{R_s^I,R_s^{II}\}$ can be achieved and the proof of Theorem 2 is finished.

\bibliographystyle{IEEEtran}
\bibliography{references}

\end{document}